\documentclass[12pt]{article}
\usepackage[T2A]{fontenc}
\usepackage[cp1251]{inputenc}
\usepackage[english]{babel}
\usepackage{amsmath,amssymb,amsthm,amsfonts,amscd,amsxtra}
\usepackage[in]{fullpage}
\newtheorem{theorem}{Theorem}
\newcommand {\mP}{\mathcal{P}}
\def\eps{\varepsilon}
\def\bR{\mathbb{R}}
\newcommand {\ds}{\displaystyle}
\begin{document}
\begin{center}
\large{\textbf{Separation of variables  \\
in one partial integrable case of Goryachev}}
\end{center}
\begin{center}
P.E. Ryabov
\end{center}
\begin{flushright}
19.12.2010
\end{flushright}
\begin{abstract}
We show that the equations of motion in one
partial integrable case of Goryachev in the
rigid body dynamics can be separated by the
appropriate change of variables, the new
variables $x, y$ being hyperelliptic functions
of time. The natural phase variables
(components of coordinates and momenta) are
expressed via $x,y$ explicitly in elementary
algebraic functions.

\end{abstract}


\section {Introduction}
The equations of motion of a rigid body about a fixed point in the integrable case found by D.N. Goryachev have the form
\begin{equation}\label{en1_1}
\begin{array}{ll}
\ds{\dot
s_1=s_2s_3+cr_2r_3-b\frac{r_2}{r_3^3},}&\ds{\dot r_1=2s_3r_2-s_2r_3,}\\
\ds{\dot s_2=-s_1s_3+cr_1r_3+b\frac{r_1}{r_3^2},}&\ds{\dot r_2=-2s_3r_1+s_1r_3,}\\
\ds{\dot s_3=-2cr_1r_2,}&\ds{\dot
r_3=s_2r_1-s_1r_2.}
\end{array}
\end{equation}
This system can be written in the Hamiltonian
form on the space ${\Bbb R}^6({\boldsymbol s},
{\boldsymbol r})$ with the Poisson brackets
\begin{equation*}
\begin{array}{c}
 \{s_i,s_j\}=-\eps_ {i j k} s_k,
\quad \{s_i,r_j\}=-\eps _ {i j k} r_k, \quad
\{r_i,r_j\}=0,\\[3mm]
 \quad 1\leqslant
i,j,k\leqslant 3,\quad
\eps=\frac{1}{2}(i-j)(j-k)(k-i)
\end{array}
\end{equation*}
and the Hamilton function
\begin{equation}\label{en1_2}
\ds{H=\frac{1}{2}(s_1^2+s_2^2+2s_3^2)+\frac{1}{2}[c(r_1^2-r_2^2)+\frac{b}{r_3^2}].}
\end{equation}

Here ${\boldsymbol s}$ is the kinetic moment of
the body,   ${\boldsymbol r}$ is identified
with the unit vector. The force depends on two
arbitrary parameters $c$ and $b$.

The geometrical integral and the area integral of the system~(\ref{en1_1})
\begin{eqnarray}\label{en1_3}
 \ds{\Gamma=r_1^2+r_2^2+r_3^2,\quad
 L=s_1r_1+s_2r_2+s_3r_3}
\end{eqnarray}
are the Casimir functions. Hence the vector field~(\ref{en1_1}) restricted to the 4-manifold
\begin{gather*}
\ds{\mP^4=\{({\boldsymbol s},{\boldsymbol
r})\in {\bR}^6: \Gamma=1,L=0\},}
\end{gather*}
is the Hamiltonian system with two degrees of freedom. To be Liouville integrable it needs, in addition to the Hamiltonian $H$, one more independent integral.

Consider the functions
\begin{equation*}\label{en1_4}
F=\left(s_1^2-s_2^2+cr_3^2-\frac{b(r_1^2-r_2^2)}{r_3^2}\right)^2+
4\left(s_1s_2-\frac{br_1r_2}{r_3^2}\right)^2\quad
\text{(D.\,N.~Goryachev \cite{bib1})}
\end{equation*}
and
\begin{equation}\label{en1_5}
\ds{K=\left(s_1^2+s_2^2+\frac{b}{r_3^2}\right)^2+2cr_3^2(s_1^2-s_2^2)+c^2r_3^4\quad
\text{(A.\,V.~Tsiganov \cite{bib2})}.}
\end{equation}
Each of them is the additional integral of (\ref{en1_1}) on the symplectic manifold $\mP^4$. According to the Liouville--Arnold theorem any regular level of the first integrals
\begin{equation}\label{s1n8}
\ds{\{({\boldsymbol s},{\boldsymbol r})\in
{\mP}^4: H=h,\quad K=k\}}
\end{equation}
is a union of two-dimensional tori bearing quasi-periodic trajectories.
The integral $K$ is not new because of the functional relation on $\mP^4$
\begin{equation*}\label{en1_6}
\ds{K=F+4bH-b^2.}
\end{equation*}

Note that introducing some terms linear in $s$ and
$r$ in the Hamiltonian function also leads to the integrable system on $\mP^4$. The additional integral in the most general form for this case is pointed out in \cite{bib3}. For such generalization no explicit integration is found yet.

For the system (\ref{en1_1}) and $b=0$ the separation of variables is found by S.A. Chaplygin \cite{bib4} leading to elliptic quadratures. In the work \cite{bib2} using the ideas of bi-Hamiltonian approach some variables of separation are suggested. Nevertheless, the corresponding equations of Abel--Jacobi type given in \cite{bib2} are not written in the explicit form. The dependencies of the phase variables on the proposed separation variables are not found either. Therefore the results of \cite{bib2} are still far from being complete.

In this work we present the explicit separation of variables in the Goryachev case. This solution does not need to involve any far-going mathematical theories and is based on the pure obvious calculation following the works of S.A. Chaplygin. We obtain the simple standard form of the separated equations and express all phase variables via the separation variables $x$ and $y$ in algebraic way.

\section{The separation of variables}
Note that by the appropriate choice of the measurement units and the directions of the moving axes one can always get $c=1$.
Let
\begin{equation}\label{s2n1}
\begin{array}{l}
\ds{u=s_1^2+s_2^2+\frac{b}{r_3^2},\quad
 z=r_3^2.}
 \end{array}
\end{equation}
These variables are introduced similar to the work \cite{bib4} where for the case $b=0$ they lead to the separation. From (\ref{en1_2}), (\ref{en1_3}) and
(\ref{en1_5}) we find the parametric equations of the integral manifold (\ref{s1n8})
\begin{equation}\label{s2n2}
\begin{array}{l}
\ds{s_1^2=\frac{(k-2b)-(u-z)^2}{4z},\quad
s_2^2=\frac{(u+z)^2-(k+2b)}{4z},\quad
s_3^2=\frac{B+\sqrt{B^2-AC^2}}{4A},}\\[3mm]
\ds{2r_1^2=2h-u-2s_3^2+1-z,\quad
2r_2^2=1-z-2h+u+2s_3^2,\quad r_3^2=z.}
 \end{array}
\end{equation}
Here
\begin{equation*}\label{s2n3}
\begin{array}{l}
\ds{A=b^2-2buz+z^2k,}\\[3mm]
\ds{B=-2\,{z}^{2}ku+2\,{z}^{2}kh+2\,{z}^{2}{u}^{2}h+3\,b{u}^{2}z+kuz+kzb-8\,
uzhb-}\\[3mm]
\ds{-{z}^{2}b+{z}^{3}b+{z}^{3}u-2\,{z}^{4}h-kb+{u}^{2}b-{u}^{3}z+4\,{b
}^{2}h-2\,{b}^{2}u,}\\[3mm]
\ds{C=-k+{u}^{2}+{z}^{2}-4\,uzh+kz+4\,bh-2\,bu+{u}^{2}z-{z}^{3}.}
\end{array}
\end{equation*}
For the existing of real solutions it is necessary to simultaneously have
\begin{equation}\label{s2n4}
\begin{array}{l}
\ds{k\geqslant 4bh-b^2,\quad k\geqslant 2b.}
\end{array}
\end{equation}
In view of (\ref{s2n4}), transform the value $B^2-AC^2$:
\begin{equation*}
\begin{array}{l}
\ds{D(u)=B^2-AC^2=z^2(z_1-u)(u-z_2)(u-z_3)(u-z_4)(z_5-u)(u-z_6).}
\end{array}
\end{equation*}
Here the roots $z_k$ of the polynomial $D(u)$
are defined as
\begin{equation*}\label{s2n5}
\begin{array}{l}
\ds{z_{1,2}=z\pm\sqrt{k-2b},\quad
z_{3,4}=-z\pm\sqrt{k+2b},}\\[3mm]
\ds{z_{5,6}=(1-z)(b\pm\sqrt{b^2+k-4bh})+2zh.}
\end{array}
\end{equation*}

 The discriminant set of $D(u)$ is
\begin{equation}\label{s2n55}
k=\pm 2b,\quad k=4hb-b^2, \quad k=(2h\mp
1)^2\pm 2b.
\end{equation}

The accessible region on the $(z,u)$-plane is then bounded by the segments of the straight lines $u = z_k$.
The quadrangle structure of the motion possibility regions usually leads to some exact separation of variables (see
\cite{bib6}, \cite{bib5}).

Consider the expression of the coefficient $A$. By virtue of (\ref{en1_5}) and (\ref{s2n1}) we have
\begin{equation*}
\ds{A=z^2[(s_2-r_3)^2+s_1^2][s_1^2+(s_2+r_3)^2]\geqslant
0}
\end{equation*}
Hence
\begin{equation}\label{s2n6}
\ds{kz^2-2buz+b^2=\xi^2.}
\end{equation}
This equation defines in space ${\Bbb R}^3(u,z,\xi)$ the second order surface which is
\textit{a one-sheet hyperboloid} and therefore has two families of rectilinear generators.
Take the parameters of these families $(x,y)$ for new variables.
The parametric equations of the hyperboloid (\ref{s2n6}) take the form
\begin{equation}\label{s2n7}
z=\frac{2b}{x+y},\quad u=\frac{xy+k}{x+y},\quad
A=\xi^2=\frac{b^2(x-y)^2}{(x+y)^2}.
\end{equation}

Denote the polynomials
\begin{equation}\label{s2n8}
\begin{array}{l}
\ds{w_1=2b+k-x^2,\quad w_2=2b-k+x^2,\quad
w_3=4bh-k-2bx+x^2,}\\[3mm]
\ds{w_4=2b+k-y^2,\quad w_5=-2b+k-y^2,\quad
w_6=-4bh+k+2by-y^2.}
\end{array}
\end{equation}
According to the above notation
\begin{equation}\label{s2n9}
\begin{array}{l}
\ds{B+\xi C=\frac{2bw_1w_2w_6}{(x+y)^4},\quad
B-\xi C=\frac{2bw_3w_4w_5}{(x+y)^4}.}
\end{array}
\end{equation}
We see then that from (\ref{s2n2}) and (\ref{s2n7})--(\ref{s2n9}) we easily obtain the algebraic expression for $s_3$:
\begin{equation*}
\ds{s_3=\frac{\sqrt{w_1w_2w_6}+\sqrt{w_3w_4w_5}}{2\sqrt{b}(x^2-y^2)}.}
\end{equation*}
In a similar way we can simplify the square roots defining other phase variables in (\ref{s2n2}). Finally we come to the following statement.
\begin{theorem} On the common level of the integrals $(\ref{s1n8})$ all phase variables
$({\boldsymbol s}, {\boldsymbol r})$
are algebraically expressed in terms of $x$ and $y$ in the form
\begin{equation}\label{s2n10}
\begin{array}{l}
\ds{s_1=-\frac{\sqrt{w_2w_5}}{2\sqrt{2b}\sqrt{x+y}},\quad
s_2=\frac{\sqrt{w_1w_4}}{2\sqrt{2b}\sqrt{x+y}},\quad
s_3=\frac{\sqrt{w_1w_2w_6}+\sqrt{w_3w_4w_5}}{2\sqrt{b}(x^2-y^2)},}\\[5mm]
\ds{r_1=\frac{\sqrt{w_2w_3w_4}-\sqrt{w_1w_5w_6}}{2\sqrt{b}(x^2-y^2)},\quad
r_2=-\frac{\sqrt{w_2w_4w_6}+\sqrt{w_1w_3w_5}}{2\sqrt{b}(x^2-y^2)},\quad
r_3=\sqrt{\frac{2b}{x+y}}.}
\end{array}
\end{equation}
\end{theorem}

To complete the separation it is necessary to derive the differential equations for the variables $x$ and $y$.
\begin{theorem}
The derivatives of the introduced variables $x$, $y$
by virtue of the system~$(\ref{en1_1})$ satisfy the equations
\begin{equation}\label{s2n11}
\ds{(x-y)\frac{dx}{dt}=-\frac{1}{\sqrt{b}}\sqrt{-W(x)},\quad
(x-y)\frac{dy}{dt}=\frac{1}{\sqrt{b}}\sqrt{-W(y)},}
\end{equation}
where
\begin{equation*}\label{s2n20}
W(s)=(s^2-k-2b)(s^2-k+2b)(s^2-2bs+4bh-k).
\end{equation*}
\end{theorem}
\begin{proof}
On one hand we have
\begin{equation}\label{s2n12}
\begin{array}{l}
\ds{\dot u=\{u,H\}=2r_3(s_2r_1+s_1r_2),\quad
\dot z=\{z,H\}=2r_3(s_2r_1-s_1r_2),}
\end{array}
\end{equation}
and on the other hand we obtain
\begin{equation}\label{s2n13}
\begin{array}{l}
\ds{\dot u=\frac{\dot x(y^2-k)+\dot
y(x^2-k)}{(x+y)^2},\quad \dot z=-\frac{2b(\dot
x+\dot y)}{(x+y)^2}.}
\end{array}
\end{equation}
From (\ref{s2n12}) and (\ref{s2n13}) we find
the expressions for $\dot x$ and $\dot y$:
\begin{equation}\label{s2n14}
\begin{array}{l}
\ds{\dot
x=-\frac{r_3(x+y)}{b(x-y)}[(s_2r_1-s_1r_2)(x^2-k)+2b(s_2r_1+s_1r_2)],}\\[5mm]
\ds{\dot
y=\frac{r_3(x+y)}{b(x-y)}[(s_2r_1-s_1r_2)(y^2-k)+2b(s_2r_1+s_1r_2)].}
\end{array}
\end{equation}
Substituting (\ref{s2n10}) into (\ref{s2n14}),
after some simple transformations we come to (\ref{s2n11}).
\end{proof}

The separated equations can also be written as the Abel--Jacobi equations
\begin{equation*}\label{s2n15}
\begin{array}{l}
\ds{\frac{dx}{\sqrt{-W(x)}}+\frac{dy}{\sqrt{-W(y)}}=0,\quad
\frac{xdx}{\sqrt{-W(x)}}+\frac{ydy}{\sqrt{-W(y)}}=-\frac{dt}{\sqrt{b}}.}
\end{array}
\end{equation*}

From algebraic point of view the flow of the Hamiltonian $H$ linearizes on the Jacobian
$J(\Gamma)$ of the hyperelliptic curve of genera $2$ given by the equation
\begin{equation*}
\Gamma: \tau^2=W(s).
\end{equation*}
The discriminant set of $W(s)$, naturally,
coincides with (\ref{s2n55}).

Thus, the formulas (\ref{s2n10}) and
(\ref{s2n11}) give the analytical real separation of variables which, in particular, provides the possibility to completely investigate the phase topology of the system including the description of the Liouville tori bifurcations.

The author is grateful to prof. M. Kharlamov for valuable discussions and advices.

\end{document}